\documentclass[letterpaper,10pt,conference]{ieeeconf} 

\IEEEoverridecommandlockouts                              
\usepackage[none]{hyphenat}
\usepackage{amsmath}
\usepackage{amssymb}
\usepackage{mathrsfs}
\usepackage{cite}
\usepackage{graphicx}
\usepackage{amsfonts}
\usepackage{amstext}
\usepackage{float}
\usepackage{fullpage}
\usepackage{nomencl}

\usepackage{color}

\newtheorem{theorem}{Theorem}[section]
\newtheorem{remark}{Remark}[section]
\newtheorem{lemma}{Lemma}[section]

\title{Co-Design of Watermarking and Robust Control \\ for Security in Cyber-Physical Systems}

\author{Raman Goyal, Christoforos Somarakis, Erfaun Noorani and Shantanu Rane 
\thanks{R. Goyal, C. Somarakis, E. Noorani, and S. Rane are with Palo Alto Research Center - A Xerox Company, Palo Alto, CA, USA. E. Noorani is with the Department of Electrical and Computer Engineering at the University of Maryland, College Park, MD, USA. 
{ \tt\small \{rgoyal,somarakis,enoorani,srane\}@parc.com, \tt\small \{enoorani\}@umd.edu} }}

\begin{document}

\maketitle

\begin{abstract}
This work discusses a novel framework for simultaneous synthesis of optimal watermarking signal and robust controllers in cyber-physical systems to minimize the loss in performance due to added watermarking signal and to maximize the detection rate of the attack. A general dynamic controller is designed to improve system performance with respect to the $\mathcal H_2$ norm, while a watermarking signal is added to improve security performance concerning the detection rate of replay attacks. The attack model considered in the paper is a replay attack, a natural attack mode when the dynamics of the system is unknown to the attacker. 
The paper first generalizes the existing result on the detection rate of $\chi^2$ detector from a static-LQR controller to a general dynamic controller. The design improvements on both robustness and security fronts are obtained by iteratively solving the convex subsets of the formulated non-convex problem in terms of the controller and watermarking signal. A semi-definite programming optimization is formulated using Linear Matrix Inequality (LMI) results to solve the larger system-level design optimization problem. We highlight the effectiveness of our method over a simplified three-tank chemical system. 
\end{abstract}

\section{Introduction}
Cyber-Physical Systems (CPS) and Internet of Things (IoT) devices have become an essential part of today's critical infrastructures, financial markets, governments, as well as our daily lives. The inherent vulnerability of these technologies to unauthorized use, digital theft, or other malicious attacks declared the area of Cyber-Security of vital importance in our modern way of living. This importance has been magnified by the proliferation of notorious malicious actors and high-profile Cyber-Attacks, especially in industrial control systems \cite{cornelius21}. The security of these systems has attracted the interest of several research communities with the activity that continues into the third decade \cite{patton2000issues,10.1145/3203245,VENKATASUBRAMANIAN2003293}. In the control community, the efforts have focused on fast and efficient methods of detection and isolation of misbehaviors \cite{sharma2010sensor,cardenas2008secure}. One particular problem of great interest is the design of watermarking signals in detecting malicious attacks \cite{rubio2017use}. This is a method of data authentication that infuses nominal data with untraceable and unrepeatable signals, typically noise, and has proven to be very effective against system attacks with the aid of observation replays. In these attacks, the attacker records the output signal for some amount of time and then replays it while attacking the system e.g. recording and replaying security videos \cite{SurveyReplay}. This way the system user would remain oblivious to ongoing attacks, without discrepancy between the replay signals and an appropriately designed watermark that triggers a detection mechanism. 

However, watermarking is typically fed into the plant as extra noise, degrading the performance of the system. Within a standard system dynamics framework, it can be shown that the more effective the watermarking, the stronger the performance degradation. The obvious potential consequences on the actual (typically nonlinear) dynamics, a degradation of the linearized system can cast upon, motivates our interest to look into the problem of the simultaneous design of watermarking architectures and robust controllers against replay attacks. Similar ideas of simultaneous design of controller and sensor/actuator placement have been considered in the literature to get the optimal desired performance. Most of the existing literature focuses on the Linear-Quadratic-Gaussian (LQG) controllers. We treat the LQG problem as a special case of the larger class of $\mathcal H_2$ problems \cite{bosgra2001design} and consider a synthesis of $\mathcal H_2$ robust controller for an LTI system to improve system volatility to random disturbances while at the same time we infuse a watermarking signal of similar random nature. It appears that these two system components are fundamentally competitive, while a straightforward coupling for optimal synthesis yields a non-convex optimization problem. To this end, we propose a design method, that although sub-optimal, provides a significant improvement on the overall robustness and detection rate security performance. For the $\mathcal H_2 $ control synthesis problem works most relevant to this paper are \cite{ROTEA_H2_1993,liu2006computational}. The literature on the design of watermarking signals for system security, of relevance to this paper, is vast. Due to space limitation, we refer the interested reader to  \cite{hashemi2020gain,liu2018Johansson,satchidanandan2016dynamic,ozel2017physical}, and cited works therein. 

The contributions and the organization of the paper can be laid out as follows:
Section~\S \ref{section: formulation} formulates the problem, details the attack model, and discusses the watermarking signal characteristics. 
Section~\S \ref{section: design} elaborates the system design approach by considering the design of the most general dynamic controller as a convex function of the watermarking signal. We further synthesize the Kalman filter in the LMI framework as a convex function of the watermarking signal and finally, generalize the existing result on the detection rate of $\chi^2$ detector with a general dynamic controller which was restricted to LQR static feedback controller \cite{mo2009secure}. 
Section~\S \ref{section: optimal} formulates the big non-convex optimization problem and then provides a solution by breaking down the problem into small convex subsets, iterations over which result in a locally optimal solution to the original non-convex problem.
Sections~\S \ref{sec:sim} and \S \ref{section: conclusion} give a detailed numerical example and final concluding remarks along with the future work. 



%






%

\section{Problem Formulation}\label{section: formulation}

\subsection{System Description}
Let us consider a discrete-time linear time-invariant (LTI) system described by the following state-space representation:
\begin{align}
x_{k+1} &= A x_k +B u_k + D w_k,\label{state_eqn2} \\
y_k &=C x_k + v_k,  \label{output_eqn2}
\end{align}
where $x_k \in \mathbb{R}^{n_x}$ is the state of the system at time $k$, $u_k \in \mathbb{R}^{n_u}$ is the control vector at time $k$. The initial state vector $x_0$ and the process noise at time $k$, $w_k$, are assumed to be independent random variables. In particular, $w_k \sim \mathcal{N}(\mathbf {0},W)$, $\forall k$, with $W\in \mathbb R^{n_w\times n_w}$ to be known and fixed covariance matrix. 
The output of the system $y_k \in \mathbb{R}^{n_y}$, is measured by a sensor network with noise modeled as another independent Gaussian random variable $v_k \sim \mathcal{N}(\mathbf {0},V)$, $\forall k$ with $V\in \mathbb R^{n_v\times n_v}$ to be known and fixed covariance matrix. The goal here is to design a control law to achieve the desired system performance. 

\subsection{Attack Model}
This section models the malicious activity that can be performed by the attacker to harm the control system described in the previous section, see figure~\ref{fig:system}. An attack model known as the replay attack is considered where the attacker has the following capabilities:

1) The attacker has access to and can record all the real-time sensor measurements for some time.

2) The attacker can replay the previously recorded data $y'_k = y_k \Longleftrightarrow y_{k-\tau}$ while attacking with a sequence of malicious control input $u^a_k \in \mathbb R^{n_{u^a}}$. 

\subsection{Optimal Watermarking Signal}
The main idea of a physical watermark is to inject a random noise (watermark signal) through some channels to excite the system and check whether the system responds to the watermark signal in accordance with the system dynamics. This section describes a novel watermarking scheme where the intensity and the channels through which watermarking signals are injected are optimized with the simultaneous design of the controller to achieve the desired system performance. 
In the literature, the watermarking signals are added to the existing actual control signal which is already designed for some desired system performance. This work simultaneously designs the control law $u_k$ and watermarking signal $\Delta u_k$. The state dynamics of the system in the presence of watermarking signal can be written as:
\begin{align}
x_{k+1} &= A x_k +B u_k +B \Delta u_k  + D w_k,\label{eq:waterSys}\\
y_k &=C x_k + v_k,
\end{align}
which assumes that the watermarking signal enters through the control input channels. Let us model the watermarking signal $\Delta u_k$ as IID Gaussian random variable with zero mean and covariance $U>O$, i.e. $\Delta u_k \sim \mathcal{N}(\mathbf{0},U)$, $\forall k$. Let us also define the inverse of covariance of the watermarking signal as:
\begin{gather}
\Gamma \triangleq U^{-1},
\end{gather}
and now, one of the goals of this research work is to find the optimization variable $\Gamma$ denoting the strength of the watermarking signal which can also specify the channels through which the signal should be added.


\section{The system design architecture}\label{section: design}
\begin{figure}[h]
    \centering
    \includegraphics[width=8cm,height=6cm,keepaspectratio]{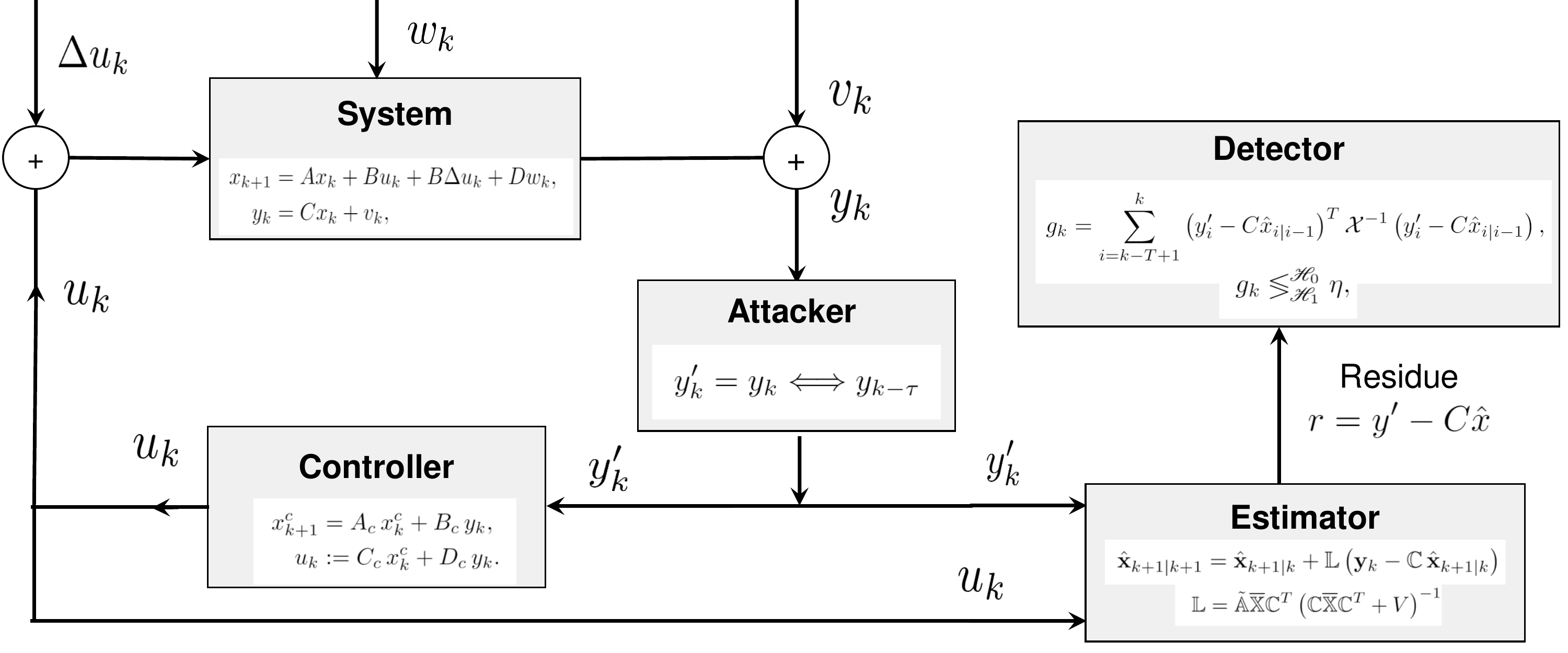} 
    \caption{Control System Architecture}
    \label{fig:system}
\end{figure}
This section details the system design architecture with a dynamic robust feedback controller and a separate estimator/detector to achieve the robust control performance and detect the attack on the system. 
The control objective is to design a dynamic compensator of the form~\eqref{e:Cntrl_dyn1}-\eqref{e:Cntrl_dyn2} below and simultaneously find the optimum watermarking signal to obtain a given level of detection in the system and also bound the $\mathcal{H}_2$ norm of the system output, i.e., $\|y\|_{\mathcal{H}_2} \leq \bar{y}_{\mathcal{H}_2}$ or bound the covariance of output, i.e., $\mathbb{E}_{\infty}(yy^T) < \bar{Y} $ for given  $\bar{Y}$ and $\bar{y}_{\mathcal{H}_2}$.

\subsection{Controller Design for Robust Performance}
Let us assume the dynamic controller of the form:
\begin{align}
x^{c}_{k+1} &= A_c\, x^{c}_{k} +B_c\, y_k, \label{e:Cntrl_dyn1}\\
u_k &:= C_c\, x^{c}_k + D_c\, y_k. \label{e:Cntrl_dyn2}
\end{align}
A standard result in the literature requires the direct feedforward term in the dynamic controller to be zero for the bounded control input covariance, thus, we will assume $D_c = O$ for the rest of the analysis\footnote{Notation ``$O$" (``$I$") is used for zero (unit) matrix of appropriate dimensions.}. Using the above compensator, the closed-loop system dynamics can be written using the augmented state vector $\mathbf{x}^T := \begin{bmatrix} x^T & {x^c}^T  \end{bmatrix}$ with augmented process noise $\mathbf{w}^T := \begin{bmatrix}  \Delta u^T & w^T & v^T \end{bmatrix}$ as:
\begin{align}
    \mathbf{x}_{k+1} &= \mathbb A \, \mathbf x_k +\mathbb {B}\, \mathbf{w}_{k}, \label{eq:closedA}\\
\mathbf{y}_{k} &= \mathbb C\, \mathbf{x}_{k} + \mathbb D\, \mathbf{w}_{k},\label{eq:closedC}
\end{align}
where 
\begin{align}
    \mathbb A &= \begin{bmatrix} A  & BC_c \\ B_c C & A_c \end{bmatrix}, & 
    \mathbb B &= \begin{bmatrix} B & D & O\\ O & O & B_c \end{bmatrix}, \\
    \mathbb C &= \begin{bmatrix} C & O \end{bmatrix},  & 
    \mathbb D & = \begin{bmatrix} O & O & I \end{bmatrix},
\end{align}
and $\mathbf w_{k}\sim \mathcal N(\mathbf 0, \mathbb W )$, where $\mathbb W:=\text{Diag}\{U,W,V\}$. It is a standard result that the above closed loop system is stable and a steady-state state covariance matrix ($\mathbb X > 0 $) exists, if:
\begin{equation} 
\mathbb A\mathbb X\mathbb A^T+\mathbb B\mathbb W\mathbb B^T < \mathbb X, \label{stable_ineq21}
\end{equation}
which again can be written as:
\begin{equation} 
\mathbb X - \mathbb A\mathbb X\mathbb X^{-1} \mathbb X^T\mathbb A^T - \mathbb B\mathbb W\mathbb B^T > O. \label{stable_ineq22}
\end{equation}

Applying Schur's complement on (\ref{stable_ineq22}) gives: 
\begin{equation}
\begin{bmatrix} 
\mathbb X & \mathbb A\mathbb X & \mathbb B \\
(\cdot)^{T}& \mathbb X & O \\
(\cdot)^{T} & (\cdot)^{T} & \mathbb W^{-1} \end{bmatrix} > O, \label{stable_ineq23}  
\end{equation}
where $(\cdot)^{T}$ represents the corresponding transpose of the symmetric block and $\mathbb W$ is the covariance matrix corresponding to $\mathbb W$, the inverse of which is:
\begin{align}
    \mathbb W^{-1} = \begin{bmatrix} \Gamma & O & O \\ O & W^{-1} & O \\ O & O & V^{-1} \end{bmatrix}.
\end{align}
\noindent It is straightforward to show that the output covariance can be bounded as:
\begin{equation}
\mathbb C\mathbb X \mathbb C^T +\mathbb D\mathbb W\mathbb D^T < \bar{Y},
\end{equation}
which can be written as:
\begin{equation}
    \begin{bmatrix}
   \bar{Y} & \mathbb C \mathbb X & \mathbb D \\
(\cdot)^{T} & \mathbb X& O \\
(\cdot)^{T} & (\cdot)^{T} & \mathbb W^{-1} 
    \end{bmatrix} > O .\label{e:out_bound}
\end{equation}
\noindent Notice that the constraint in Eqn.~(\ref{stable_ineq23}) is not an LMI. We need to perform congruence transformation and change of variables to convert them to LMIs \cite{Scherer_LMI_1997,Integrating2021MSSP}. Let us define and partition the matrix as:
\begin{equation*}
    \mathbb X \triangleq  \begin{bmatrix}
X & P^{T} \\
P & \hat{X}
\end{bmatrix} , \quad \mathbb X^{-1} \triangleq \begin{bmatrix}
Y & S \\
S^{T} & \hat{Y}\end{bmatrix},
\end{equation*}
and the transformation matrix $\mathbb{T}$ 
\begin{equation*}
    \mathbb{T} \triangleq \begin{bmatrix}
    I & Y \\ O & S^T 
    \end{bmatrix}
    \end{equation*}
    and associated congruence transformation $\{\mathcal T, \tilde{\mathcal T}\}$
    \begin{equation*}
    \mathcal T \triangleq \begin{bmatrix}
    \mathbb{T}  & O & O \\
O & \mathbb{T} & O \\
O & O & I
    \end{bmatrix}, \tilde{\mathcal T} \triangleq \begin{bmatrix}
    I & O & O \\
O & \mathbb{T} & O \\
O & O & I
    \end{bmatrix},
\end{equation*}
\noindent respectively. Then applying $\{\mathcal T, \tilde{\mathcal T}\}$ to Eqn.~(\ref{stable_ineq23}) and Eqn.~(\ref{e:out_bound}), we obtain:
\begin{equation}\label{eq: transformation1}
  \mathcal T^T  \begin{bmatrix}
    \mathbb X & \mathbb A \mathbb X & \mathbb B \\
(\cdot)^{T} & \mathbb X& O \\
(\cdot)^{T} & (\cdot)^{T} & \mathbb W^{-1}
    \end{bmatrix} \mathcal T > O,
\end{equation}
\begin{equation} \label{eq: transformation2}
    \tilde{\mathcal T}^T \begin{bmatrix}
   \bar{Y} & \mathbb C \mathbb X & \mathbb D \\
(\cdot)^{T} & \mathbb X& O \\
(\cdot)^{T} & (\cdot)^{T} & \mathbb W^{-1}
    \end{bmatrix} \tilde{\mathcal T} > O .
\end{equation} Expansion of \eqref{eq: transformation1} and \eqref{eq: transformation2} under an appropriate change of variables leads to a set of LMIs that do not depend on $S$ or $P$. Once the $X, Y$ are obtained, matrices $S$ and $P$ need to be constructed using:
%
\begin{equation}
Y X+S P=I.
\end{equation}
\noindent Notice that when the controller has the same order as the plant, $S$ and $P$ are square and non-singular matrices, in which case the controller gain matrices can be calculated as:
\begin{equation}
\begin{split}
\begin{bmatrix}
A_{c} & B_{c} \\
C_{c} & D_{c}
 \end{bmatrix} = \begin{bmatrix}
S^{-1} & -S^{-1} Y B \\
O & I
 \end{bmatrix} & \cdot \begin{bmatrix}
Q-Y A X & F \\
L & O
 \end{bmatrix} \cdot \\ & \cdot \begin{bmatrix}
P^{-1} & O \\
-C X P^{-1} & I
 \end{bmatrix} .
\end{split} \label{e:RLFQ}
\end{equation}
Also, notice that the change of variables operator is affine in $L, F$, and $Q$. 

Another constraint to consider for a feasible controller is to have a finite bou2nd on the control input covariance:
\begin{align}
\mathbb{E}_{\infty}(u_k u_k^T) < \bar{U}.
\end{align}
which uses the earlier mentioned condition:  $D_c = O$.




\subsection{Estimator Design for Fault Detection}
We will use the standard Kalman filter to estimate the system states and the widely known $\chi^{2}$ detector for fault detection in the closed-loop control system \cite{mehra1971innovations}. The Kalman filter estimate, $\hat{\mathbf{x}} = \begin{bmatrix}
\hat{x} \\\hat{x}^{c} 
\end{bmatrix}$,
for the LTI system defined in Eqn.~(\ref{state_eqn2}) can be written as:
\begin{align}
    \hat{\mathbf x}_{k+1 \mid k+1} =\hat{\mathbf x}_{k+1 \mid k}+\mathbb L \left(\mathbf y_{k+1}
    -\mathbb C\,\hat{\mathbf{x}}_{{k+1 \mid k}}\right),
\end{align}
where $\mathbb L$ is the Kalman gain given as:
\begin{align}
\mathbb L =     \tilde{\mathbb A} \overline{\mathbb X}\mathbb C^T \left(\mathbb C \overline{\mathbb X}\mathbb C^T+ V \right)^{-1}, \label{eq:KalmanG}
\end{align}
where $\overline{\mathbb X}$ is the steady-state state covariance matrix that satisfies
\begin{equation*}
\begin{split}
\overline{\mathbb X} = \tilde{\mathbb{A}} \overline{\mathbb X} \tilde{ \mathbb A}^T &+ \mathbb B\mathbb W \mathbb{B}^T - \mathbb K\, V \,\mathbb K^T - \\
&-\tilde{\mathbb A}\overline{\mathbb X}\mathbb C^T \left(\mathbb C \overline{\mathbb X}\mathbb C^T+V\right)^{-1}\mathbb C\overline{\mathbb X}\tilde{\mathbb A}^T,
\end{split}
\end{equation*}
where $\tilde{\mathbb A}=\mathbb A-\mathbb K \mathbb C$, $V=\mathbb D \mathbb W \mathbb D^T$, $ \mathbb K= \mathbb S (V)^{-1}  $, and $\mathbb S = \mathbb B \mathbb W  \mathbb D^T$. After substituting for the system matrices (Eqns.~(\ref{eq:closedA})-(\ref{eq:closedC})),  we get:
\begin{equation}
     \overline{\mathbb X} = \tilde{\mathbb A} \overline{\mathbb X}\tilde{\mathbb A}^T + \mathbb B\tilde{\mathbb W}\mathbb B^T -\tilde{\mathbb A}\overline{\mathbb X}\mathbb C^T \left(\mathbb C \overline{\mathbb X}\mathbb C^T+V\right)^{-1}\mathbb C\overline{\mathbb X}\tilde{\mathbb A}^T, \label{eq:Xecl}
\end{equation}
where
$    \tilde{\mathbb W} = \begin{bmatrix} U & O & O \\ O & W & O \\ O & O & \epsilon \end{bmatrix}, $
and $\epsilon>0$ a small enough number, added to make $\tilde{\mathbb W}$ positive definite. Now, using the definition of gain matrix $\mathbb L$ from Eqn.~\ref{eq:KalmanG}, Eqn.~\ref{eq:Xecl} can be written as:
\begin{equation*}
    \nonumber \overline{\mathbb X} = (\tilde{\mathbb A}-\mathbb L \mathbb C)  \overline{\mathbb X}(\tilde{\mathbb A}-\mathbb L \mathbb C)^T + \mathbb B\tilde{\mathbb W}\mathbb B^T + \mathbb L V \mathbb L^T.
\end{equation*} With the aim of minimizing $\operatorname{trace}(\overline{\mathbb X})$, we use Comparison Lemma \cite{Scherer_LMI_1997} to write:
\begin{equation*}
     \overline{\mathbb X} - (\tilde{\mathbb A}-\mathbb L \mathbb C)  \overline{\mathbb X}(\tilde{\mathbb A}-\mathbb L \mathbb C)^T  - \mathbb LV\mathbb L^T - \mathbb B\tilde{\mathbb W}\mathbb B^T > O,
\end{equation*}
which can be written using Schur's complement as:
\begin{align}
    \begin{bmatrix}
    \overline{\mathbb X} & \tilde{\mathbb A}-\mathbb L \mathbb C & \mathbb L & \mathbb B \\
    (\cdot)^T & \overline{\mathbb X}^{-1} &O&O\\
    (\cdot)^T & (\cdot)^T & V^{-1} &O \\
    (\cdot)^T & (\cdot)^T & (\cdot)^T & \tilde{\mathbb W}^{-1} 
    \end{bmatrix} > O.
\end{align}

Let us define $\overline{\mathbb P}^{-1}  = \overline{\mathbb X}$ and multiply both sides by matrix $[\operatorname{blkdiag}(\overline{\mathbb P},I,I,I)]$ to apply congruence transformation:
\begin{align}
   \begin{bmatrix}
    \overline{\mathbb P} & \overline{\mathbb P}\tilde{\mathbb A}-\overline{\mathbb P}\mathbb L \mathbb C & \overline{\mathbb P}\mathbb L & \overline{\mathbb P}\mathbb B \\
    (\cdot)^T & \overline{\mathbb P} &O&O\\
    (\cdot)^T & (\cdot)^T & V^{-1} &O \\
    (\cdot)^T & (\cdot)^T & (\cdot)^T & \tilde{\mathbb W}^{-1} 
    \end{bmatrix} > O,
\end{align}

Now defining $\overline{\mathbb Y} = \overline{\mathbb P}\mathbb L$, we obtain the following LMI in $\overline{\mathbb P},\overline{\mathbb Y}$ and $\Gamma$ as:
\begin{align}
   \begin{bmatrix}
    \overline{\mathbb P} & \overline{\mathbb P}\tilde{\mathbb A}-\overline{\mathbb Y} \mathbb C & \overline{\mathbb Y} & \overline{\mathbb P}\mathbb B \\
    (\cdot)^T & \overline{\mathbb P} &O&O\\
    (\cdot)^T & (\cdot)^T & V^{-1} &O \\
    (\cdot)^T & (\cdot)^T & (\cdot)^T & \tilde{\mathbb W}^{-1}
    \end{bmatrix} > O.\label{eq:PL}
\end{align}
The above equation allows us to write the synthesis of the Kalman filter in the LMI framework as a convex function of optimization variable $\Gamma$ which denotes the strength of the watermarking signal.

\subsection{Fault Detection using $\chi^{2}$ Detector}
Let us partition the steady-state state covariance matrix $\overline{\mathbb X}$ as:
\begin{equation}\label{eq: mathbbX}
    \overline{\mathbb X} \triangleq  \begin{bmatrix}
X^e & P^{e^T} \\
P^e & \hat{X}^e
\end{bmatrix}.
\end{equation} 
\noindent The $\chi^{2}$ detector for fault detection is based on the following lemma \cite{mehra1971innovations}:
\begin{lemma}\label{lem: chi2} The residues with the Kalman filter estimate $y_{k}-C \hat{x}_{k \mid k-1}$ follows i.i.d. Gaussian distribution with $\mathbf{0}$ mean and covariance
$$\mathcal{X}=C X^e C^{T}+V,$$
where $X^e$ is the block sub-matrix of $\mathbb X$ from Eq.\eqref{eq: mathbbX}.
\end{lemma}
\noindent In view of the above result, the $\chi^{2}$ detector takes the following form at time $k$ :
\begin{align*}
&g_{k}=\sum_{i=k-T+1}^{k}\left(y'_{i}-C \hat{x}_{i \mid i-1}\right)^{T} \mathcal{X}^{-1}\left(y'_{i}-C \hat{x}_{i \mid i-1}\right),\\ 
\end{align*} 
subject to hypothesis testing 
\begin{equation*}
    g_{k} \lessgtr_{\mathscr{H}_{1}}^{\mathscr{H}_0} \eta,
\end{equation*} 
where $T$ is the window size of detection and $\eta$ is the threshold that is related to the false alarm rate. When the system is under normal operation, the left of the above equation is $\chi^{2}$ distributed with $m T$ degrees of freedom. Furthermore, $\mathscr{H}_{0}$ denotes the system is under normal operation while $\mathscr{H}_{1}$ denotes a triggered alarm. Here, define the probability of false alarm $\alpha_{k}$ and the probability of detection rate $\beta_{k}$ as:
\begin{align}
\alpha_{k} \triangleq \mathbb{P}\left(g_{k}>\eta \mid \mathscr{H}_{0}\right), ~~ \beta_{k} \triangleq \mathbb{P}\left(g_{k}>\eta \mid \mathscr{H}_{1}\right).
\end{align}

The following result highlights the impact of detector $g_k$ during attack and non-attack events, in the presence of a watermarking signal.
\begin{theorem}\label{thm}
For a linear time-invariant system described using Eqn.~(\ref{eq:waterSys}) and controlled with a general dynamic controller given by Eqns.~(\ref{e:Cntrl_dyn1}) and (\ref{e:Cntrl_dyn2}), the expectation of $\chi^{2}$ detector $g_k$, in the presence of an attack, converges as:
\begin{equation}
     \lim_{k\rightarrow \infty } \mathbb E[g_k] =m T+2 \operatorname{trace}\left(C^{T} \mathcal{X}^{-1} C \mathcal{U}\right)T,
\label{e:Attack}
\end{equation}
where $\mathcal{U}$ is the solution of the following Lyapunov equation:
\begin{align}
\mathcal{U}-B U B^{T}=\mathcal{A} \mathcal{U} \mathcal{A}^{T}, \label{eq:calU}
\end{align}
where $\mathcal{A} \triangleq A(I-L C)$, $\mathcal{B} \triangleq B C_c  A_c$. The expectation of $\chi^{2}$ detector, in the absence of an attack is given as:
\begin{equation}
\lim_{k\rightarrow \infty } \mathbb E[g_k] = m T.
\label{e:Noattack}
\end{equation}
\end{theorem}
\begin{proof}
Please see the appendix for the proof.
\end{proof}

\begin{remark}
Please notice in the proof that there were some traces of residual dynamics due to the dynamic controller which vanished asymptotically, to give the same steady-state expression of the detection rate as given in \cite{mo2009secure}. Hence, although the approach and the expressions turned out identical, this new result generalizes it to dynamic controllers.  
\end{remark}

Finally, the main idea behind the detector is that the replay attack deviations will avert the detection rate from converging to the false alarm rate due to the statistical independence of watermarking signals as given in Eqn.~(\ref{e:Noattack}).


\section{Optimal Watermarking Signal Design}\label{section: optimal}
This section details the formulation for the simultaneous design of watermarking and controller as a big non-convex optimization problem, the sub-optimal solution for which is found by iterating over two convexified subproblems. To better compare and present the results, the following subsections formulate three different optimization problems for maximizing the detection rate and minimizing the performance loss.

\subsection{Detection Rate by adding Random Watermarking Signal for Fixed $\mathcal{H}_2$ Controller}
This subsection details the calculation of the detection rate for an already designed closed-loop system with a $\mathcal{H}_2$ controller for a fixed but randomly chosen watermarking signal.
We first solve a $\mathcal{H}_2$ controller in the presence of watermarking signal of strength $U_{R} = 0.1 I$ as:
$$  \left.
    \begin{array}{ll}
        \textit{minimize} & \operatorname{trace}(\bar{Y})  \\
        \textit{s.t.} & \text{Eqns}.~\eqref{eq: transformation1}-\eqref{eq: transformation2}~ \text{(LMIs)}
    \end{array}
\right \} ~~~~~ \text{Prob. A-I}$$
\noindent making it a convex optimization problem for the optimization variables $A_c, B_c$, and $C_c$ (controller matrices) (ref.~Eqn.~\eqref{e:RLFQ}).
Then, a Kalman filter is designed for the closed-loop $\mathcal{H}_2$ system with the same watermarking signal by:
$$  \left.
    \begin{array}{ll}
        \textit{maximize} & \operatorname{trace}(\mathbb P)  \\
        \textit{s.t.} & \text{Eqn}.~\eqref{eq:PL} ~\text{(LMI)}
    \end{array}
\right \}  ~~~~~~~~~ \text{Prob. A-II}$$

Finally, the performance loss in the closed-loop $\mathcal{H}_2$ system - $\|y\|_{\mathcal{H}_{2_{R}}}$ and the expectation of $\chi^{2}$ detector for fixed time window $T$ - $\operatorname{trace}\left(C^{T} \mathcal{X}^{-1} C \mathcal{U}\right)$ is calculated for the randomly chosen watermarking signal.

\subsection{Maximizing Detection Rate by adding Optimal Watermarking Signal for Fixed $\mathcal{H}_2$ Controller}
In this subsection, we design the optimal watermarking signal to maximize the detection rate for the fixed $\mathcal{H}_2$ controller that was used in the previous subsection while maintaining the same performance loss in the closed-loop system. So the optimal design of watermarking signal is to solve for the optimization variable $U/\Gamma$ to: 
$$  \left.
    \begin{array}{ll}
        \textit{maximize} & \operatorname{trace}\left(C^{T} \mathcal{X}^{-1} C \mathcal{U}\right)  \\
        \textit{s.t.} & \text{Eqns}.~\eqref{eq: transformation1}, \eqref{eq: transformation2}, \eqref{eq:PL} ~\text{(LMIs)}\\
        \textit{s.t.} & \|y\|_{\mathcal{H}_{2_{R}}}>\operatorname{trace}(\bar{Y}) ~\text{(LMI)}\\
        \textit{s.t.} & \text{Eqn}.~\eqref{eq:calU} ~\text{(Non-convex)}
        \end{array}
\right \}  \text{Prob. B}$$

Notice that Eqn.~\eqref{eq:calU} has a linear relation between $\mathcal{U}$ and $U$ and gives a unique solution for $\mathcal{U}$ for fixed Kalman Gain $\mathbb{L}$. To convexify the problem, we take the constraint corresponding to Eqn.~\eqref{eq:calU} out of the optimization and solve for $\mathcal{U}$ separately afterward as it gives a unique solution for fixed Kalman Gain $\mathbb{L}$. Also, observe that the nonlinear maximization in detection rate $\mathcal{X}^{-1}$ can also be written in the linear form of $\bar{\mathbb{P}}$, i.e. convexifying it, following the definition $\overline{\mathbb X}^{-1}  = \overline{\mathbb P}$ given in Section \S \ref{section: design}.

\subsection{Maximizing Detection Rate and Performance Loss by Co-Optimal Watermarking Signal and $\mathcal{H}_2$ Controller}
The joint design of watermarking signal and $\mathcal{H}_2$ controller should result in lesser closed-loop system performance loss and a higher detection rate. This would require simultaneously solving for the watermarking signal and the controller matrices $\{U/\Gamma, A_c, B_c, C_c\}$, as the optimization variables:
$$  \left.
    \begin{array}{ll}
        \textit{maximize} & \operatorname{trace}\left(C^{T} \mathcal{X}^{-1} C \mathcal{U}\right)  \\
        \textit{s.t.} & \text{Eqns}.~\eqref{eq: transformation1}, \eqref{eq: transformation2} ~\text{(LMIs)}\\
        \textit{s.t.} & \|y\|_{\mathcal{H}_{2_{R}}}>\operatorname{trace}(\bar{Y}) ~\text{(LMI)}\\
        \textit{s.t.} & \text{Eqns}.~\eqref{eq:calU},\eqref{eq:PL} ~\text{(Non-convex)}
        \end{array}
\right \}  \text{Prob. C}$$

The above mentioned problem is nonconvex in $U/\Gamma$ and $A_c, B_c, C_c$ through Eqns.~\eqref{eq:calU}, \eqref{eq:PL}. We solve this nonconvex optimization problem by iterating over two convex subsets of the problem.
We first solve for the $A_c, B_c, C_c$ with an initial randomly chosen $U/\Gamma$ as:
$$  \left.
    \begin{array}{ll}
        \textit{minimize} & \operatorname{trace}\left(\bar{Y}\right)  \\
        \textit{s.t.} & \text{Eqns}.~\eqref{eq: transformation1}, \eqref{eq: transformation2} ~\text{(LMIs)}\\
        \textit{s.t.} & \|y\|_{\mathcal{H}_{2_{R}}}>\operatorname{trace}(\bar{Y}) ~\text{(LMI)}
        \end{array}
\right \}  \text{Prob. C-I}$$
and then, we fix the controller matrices $A_c, B_c, C_c$ to solve for the watermarking signal $U/\Gamma$ as:
$$  \left.
    \begin{array}{ll}
        \textit{maximize} & \operatorname{trace}\left(C^{T} \mathcal{X}^{-1} C \mathcal{U}\right)  \\
        \textit{s.t.} & \text{Eqns}.~\eqref{stable_ineq23}, \eqref{e:out_bound} ~\text{(LMIs)}\\
        \textit{s.t.} & \|y\|_{\mathcal{H}_{2_{R}}}>\operatorname{trace}(\bar{Y}) ~\text{(LMI)}\\
        \textit{s.t.} & \text{Eqns}.~\eqref{eq:PL} ~\text{(LMI)}
        \end{array}
\right \}  \text{Prob. C-II}$$
and the two convex optimization problems are iterated until converged to find a stationary point of the optimization. The same approximation is used for Eqn.~\eqref{eq:calU} as mentioned in the previous subsection.


\vspace{-1mm}
\section{Numerical Simulation} \label{sec:sim}
To illustrate the proposed concept of simultaneous design of feedback controller and optimal watermarking signal for the security of cyber-physical systems, we consider a chemical process shown in Figure~\ref{fig: tank} linearized along the lines of \cite{milovsevic2019estimating}. The three states $x=(x_1,x_2,x_3)^T$ describe the level of water in tanks 2 and 3 and the temperature of water in tank 2. The control inputs are two flow pumps, one valve, and one heater, as illustrated in Figure \ref{fig: tank}. 
The process noise and measurement noise is assumed to be zero-mean Gaussian with covariance $W$$=$$10^{-3}\,I$ and $V$$=$$10^{-3}\,I$, respectively. 
The linearized dynamics of this three-tank system can be written as  \cite{milovsevic2019estimating}:
$$A\hspace{-0.05in}  = \hspace{-0.05in} \begin{bmatrix}
0.96 & 0 & 0 \\
0.04 & 0.97 & 0 \\
-0.04 & 0 & 0.9
\end{bmatrix},
B \hspace{-0.05in} = \hspace{-0.05in} \begin{bmatrix}
8.8 & -2.3 & 0 & 0\\
0.2 & 2.2 & 4.9 & 0 \\
-0.21 & -2.2 & 1.9 & 21
\end{bmatrix}$$
and $D=I$. We assume all the states are available for measurement with noisy sensors, i.e., $C=I$.

\begin{figure}[ht!]
    \centering
    \includegraphics[scale=0.6]{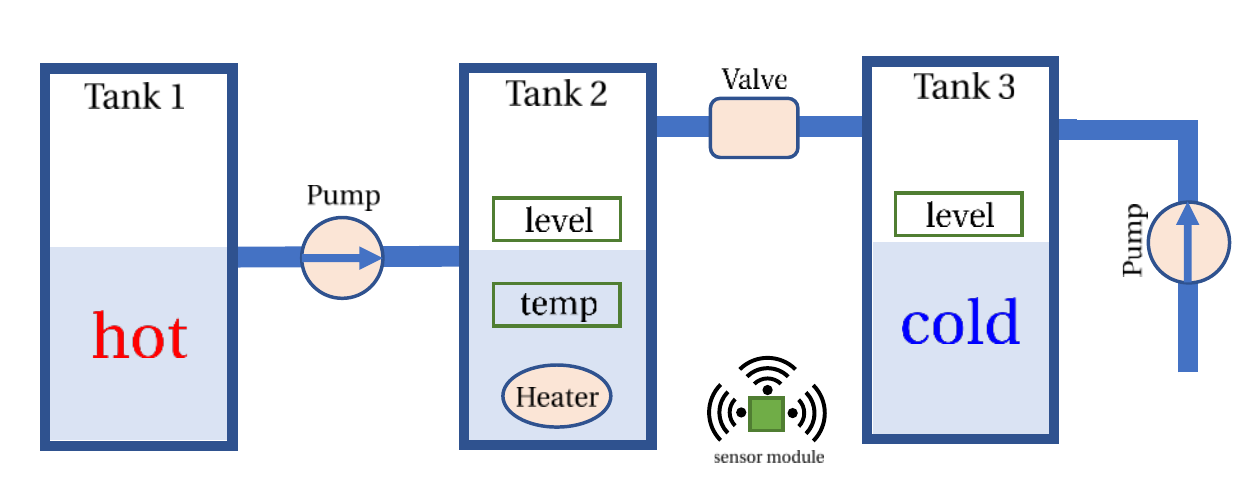}
    \caption{\small{The chemical process with four actuators (hot, cold pumps, valve and heater) that control the level and temperature of tank 2 and level of tank 3.}}
    \label{fig: tank}
\end{figure}

Figure~\ref{f:H20} shows the detection rate for the replay attack with attack starting at 11th step (replying the first 10 steps at time step 11) for the fixed $\mathcal{H}_2$ controller designed to minimize the $\|y\|_{\mathcal{H}_2}$ norm as mentioned in \S IV.A for different watermarking signal strengths. Notice that there is no detection till the 10th time step and then the detection rate increases with increased noise of the watermarking signal.
    \vspace{-2mm}
\begin{figure}[h!]
    \centering
    \includegraphics[width=0.4\textwidth]{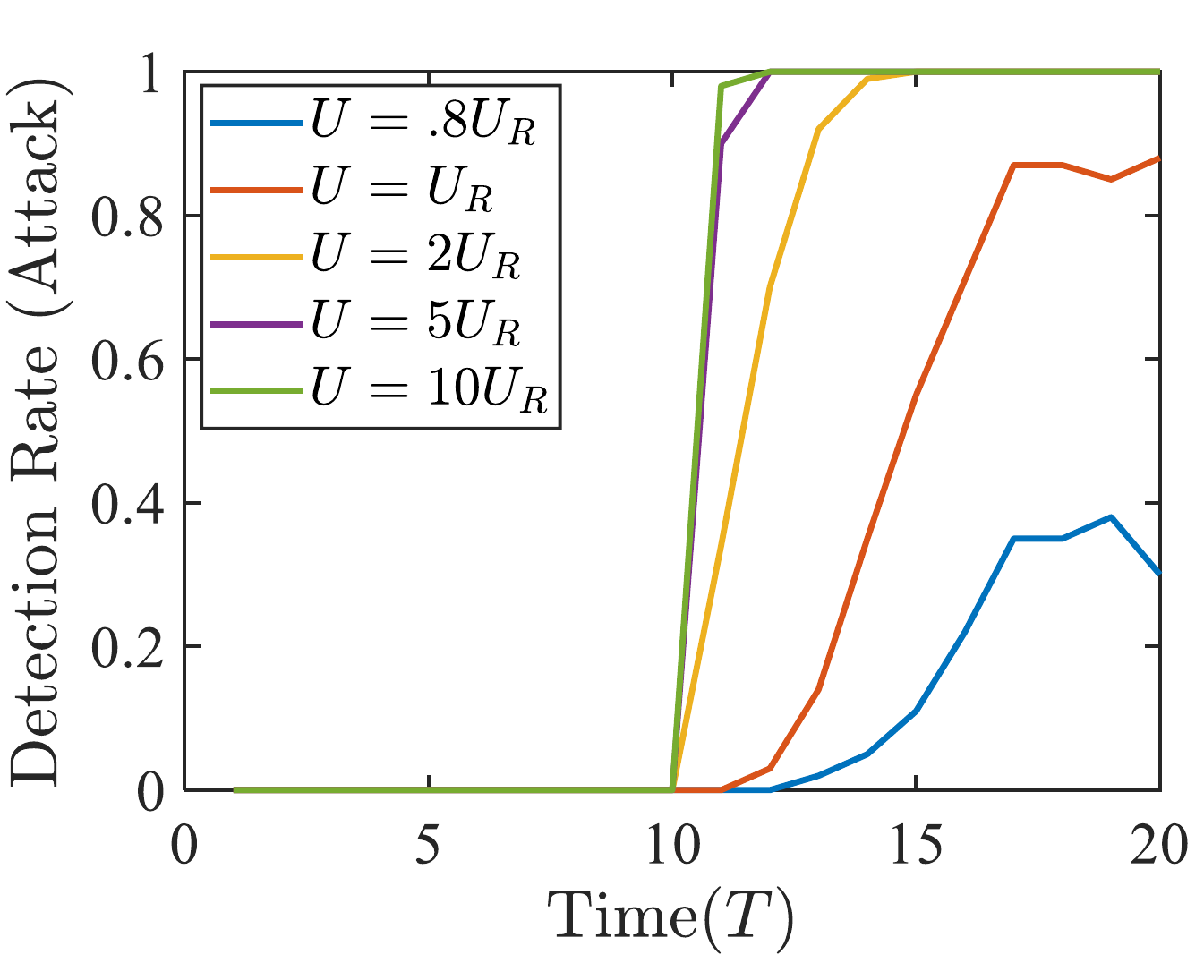} 
    \vspace{-2mm}
    \caption{Detection rate for different watermarking signal strength with the fixed $\mathcal{H}_2$ controller.}
    \label{f:H20}
\end{figure}
\vspace{-4mm}
\begin{table}[h!]
\caption{Table containing optimal signals for different subsections IV.A, IV.B and IV.C. }
\begin{tabular}{|l|l|l|l|l}
\cline{1-4}
 &        Watermarking - $U$                         &  $\|y\|_{\mathcal{H}_{2}}$    &  Detec.    &  \\ \cline{1-4}
$U=0$ & diag({[}0;0;0;0{]})             & 0.01 & 0.00    &  \\ \cline{1-4}
$U=U_{R}$ & diag({[}0.1;0.1;0.1;0.1{]})         & 111  & 3.12 &  \\ \cline{1-4}
$U=U_{Opt}$ & diag({[}.098;.187;.134;.094{]}) & 110  & 3.18 &  \\ \cline{1-4}
$U=U_{Opt+Ctrl}$ & diag({[}.097;.175;.092;.108{]}) & 111  & 3.36 &  \\ \cline{1-4}
\end{tabular}
\end{table}

\begin{figure}[ht!]
    \centering
    \includegraphics[width=0.45\textwidth]{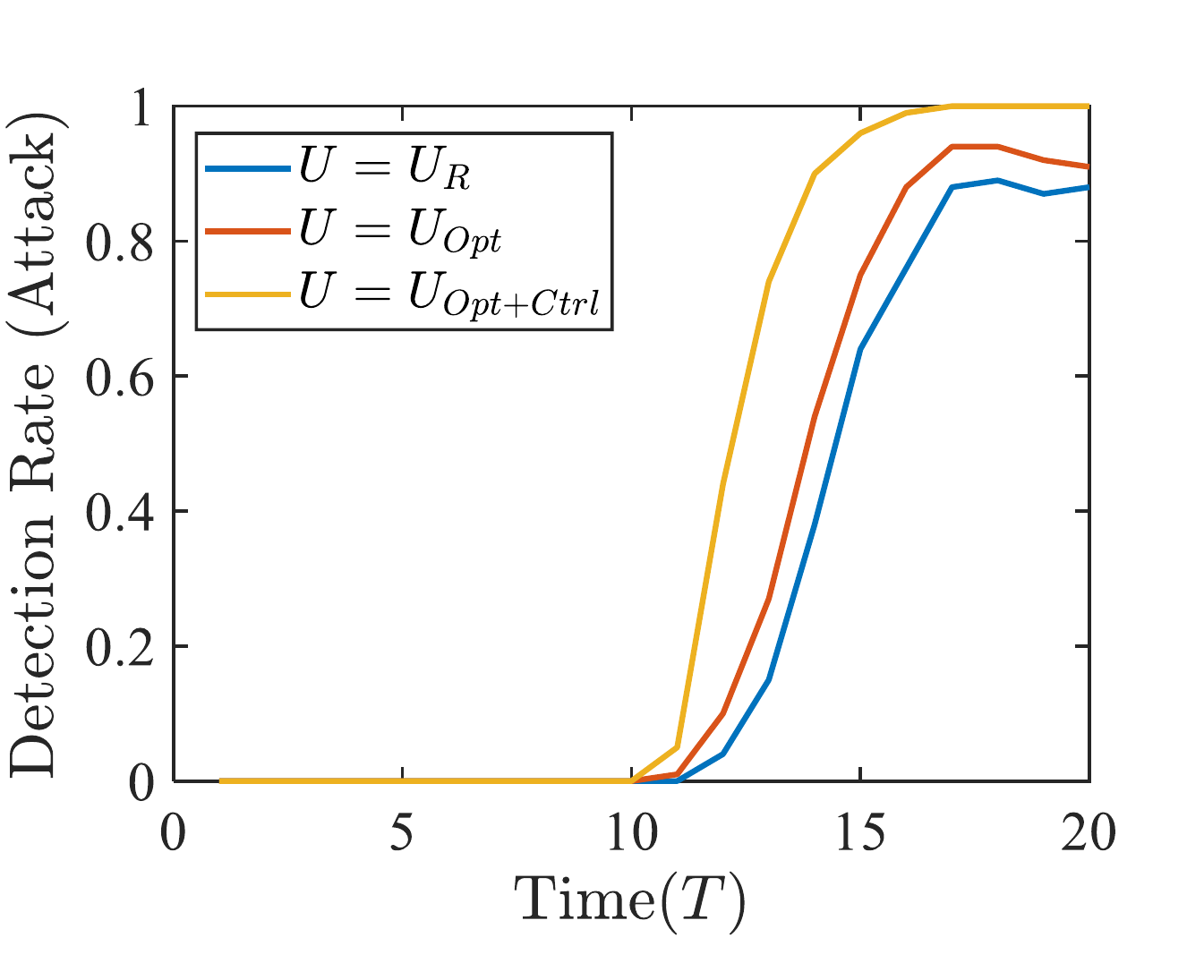} 
    \vspace{-2mm}
    \caption{Detection rate for three different cases mentioned in \S IV.A, IV.B and IV.C.}
    \label{f:OptH20}
\end{figure}

Table I gives the watermarking signal $U$ and the detection rate (i.e. $\operatorname{trace}\left(C^{T} \mathcal{X}^{-1} C \mathcal{U}\right)$) for three different cases mentioned in \S IV.A, IV.B and IV.C. Notice that the simulation results are designed to match the $\mathcal{H}_2$ of the output for all the three cases while maximizing the detection rate. The detection rate is zero and the performance is optimal in the absence of the watermarking signal $U=0$. A detection is observed for the randomly chosen watermarking signal $U=U_R$ with a significant loss in performance. The results show that the detection rate is improved with optimized watermarking signal $U_{Opt}$ with the same performance loss in the system, and it is improved further by co-optimizing the watermarking signal with the dynamic controller $U_{Opt+Ctrl}$.
Figure~\ref{f:OptH20} shows the detection rate for three different cases with reference watermarking signal $U_R$, optimized watermarking signal $U_{Opt}$, and co-optimized watermarking signal with controller $U_{Opt+Ctrl}$. Notice that the optimal pair of watermarking signal and controller results in faster and better detection of the replay attack.


\section{Conclusion}\label{section: conclusion}
This paper developed a novel system-level design approach by the simultaneous selection of dynamic controller parameters and the architecture of watermarking signal to get the desired performance and detect the attack on the system. 
We explored how a watermarking signal generator can be augmented to a $\mathcal H_2$ controlled system to detect reply attacks when there should be no change in the $\mathcal H_2$ legacy controller, and investigated how the simultaneous design/re-configuration of the $\mathcal H_2$ controller and watermarking signal would benefit the effectiveness of the watermarking signal. The design of the optimal watermarking signal for a given controller and detection rate, and the co-design of the controller and the watermarking signal for a given detection rate led to non-linear optimization problems which we attempt to approximately solve using iterative approaches. 
A follow-up work would be to simultaneously design the sensor and watermarking architecture to minimize the loss in performance due to added watermarking signal and still increase the detection rate of the attack.
\vspace{-1mm}
\bibliographystyle{IEEEtran}
    \vspace{-2mm}
\bibliography{Refs}
    \vspace{-2mm}
\vspace{-2mm}
\appendix
\vspace{-2mm}
\textbf{Proof for Theorem \ref{thm}:}
Let us assume $y_{k}^{\prime}$ as the output of the following virtual system  \cite{mo2009secure}:
\begin{align}
x_{k+1}^{\prime} &=A x_{k}^{\prime}+B u_{k}^{\prime}+w_{k}^{\prime} + B \Delta u^{\prime}_k, \\
y_{k}^{\prime}& =C x_{k}^{\prime}+v_{k}^{\prime},
\\\hat{x}_{k+1 \mid k}^{\prime} &=A \hat{x}_{k \mid k}^{\prime}+B u_{k}^{\prime}+ B \Delta u^{\prime}_k,\\ \hat{x}_{k+1 \mid k+1}^{\prime} &=\hat{x}_{k+1 \mid k}^{\prime}+L\left(y_{k+1}^{\prime}-C\hat{x}_{k+1 \mid k}^{\prime}\right), 
\end{align}
with initial conditions $x_{0}^{\prime}$ and let the dynamic controller for the virtual system :
\begin{align}
x_{k+1}^{\prime,c} &= A_c x^{\prime,c}_{k} +B_c y^{\prime}_k, \label{e:Cntrl_dyn_vir} \\
u^{\prime}_k &= C_c x^{\prime,c}_{k} + D_c y^{\prime}_k.
\end{align}

Basically, to do a replay attack, we assume that the attacker runs the virtual system and records the sequence $y_{k}$ from time $t$. Then the virtual system is just a time shifted version of the real system, with $x_{k}^{\prime}=x_{t+k}, \hat{x}_{k \mid k}^{\prime}=\hat{x}_{t+k \mid t+k}$.

Now we rewrite the estimation of the Kalman filter $\hat{x}_{k \mid k-1}$ in the following recursive way to be able to detect the attack using the $\chi^{2}$ detector:
\begin{align*}
\hat{x}_{k+1 \mid k} &=A \hat{x}_{k \mid k}+B u_{k} + B \Delta u_k \\ &=A \hat{x}_{k \mid k}  + B C_c  x_k^c + B \Delta u_k,\\
\nonumber &=A\left[\hat{x}_{k \mid k-1}+L\left(y_{k}^{\prime}-C \hat{x}_{k \mid k-1}\right)\right] \\ &~~~~~~+ B C_c  (A_c x_{k-1}^c + B_c y_{k}^{\prime}) + B \Delta u_k,\\
\nonumber &=A(I-L C) \hat{x}_{k \mid k-1}+
B C_c  A_c x_{k-1}^c \\ &~~~~~~+ (B C_c B_c - A L) y_{k}^{\prime} + B \Delta u_k.
\end{align*}

It can be easily shown that the same equation holds true for the virtual system $\hat{x}_{k \mid k-1}^{\prime}$ :
\begin{align*}
\nonumber \hat{x}_{k+1 \mid k}^{\prime}&=A(I-L C) \hat{x}_{k \mid k-1}^{\prime}+
B C_c  A_c x_{k-1}^{\prime,c} \\ &~~~~~~~+ (B C_c B_c - A L) y_{k}^{\prime} + B \Delta u^{\prime}_k.
\end{align*} Subtracting the above two equations, define $\mathcal{A} \triangleq A(I-L C)$, $\mathcal{B} \triangleq B C_c  A_c,$ to get:
\begin{align*}
\hat{x}_{k+1 \mid k}-\hat{x}_{k+1 \mid k}^{\prime}&=\mathcal{A} \left(\hat{x}_{k \mid k-1}-\hat{x}_{k \mid k-1}^{\prime}\right) \\ & + \mathcal{B} (x_{k-1}^c- x_{k-1}^{\prime,c} ) +  B (\Delta u_k - \Delta u^{\prime}_k),
\end{align*}
which can also be written for $\hat{x}_{k \mid k-1}-\hat{x}_{k \mid k-1}^{\prime}$ as the summation of the terms from initial time-steps as:
\begin{align}
\nonumber \hat{x}_{k \mid k-1}-\hat{x}_{k \mid k-1}^{\prime}&=\mathcal{A}^{k}\left(\hat{x}_{0 \mid-1}-\hat{x}_{0 \mid-1}^{\prime}\right) \\
\nonumber + &~~~ \sum_{i=1}^{k} \mathcal{A}^{k-i} \mathcal{B} (x^c_{i-1}- x_{i-1}^{\prime,c} ) \\+ &~~~ \sum_{i=1}^{k} \mathcal{A}^{k-i}  B (\Delta u_i - \Delta u^{\prime}_i). \label{e:diff_real_vir}
\end{align}

Subtracting Eqn.~(\ref{e:Cntrl_dyn1}) and Eqn.~(\ref{e:Cntrl_dyn_vir}) for the case of replay attack with $y_k = y^{\prime}_k$ gives:
\begin{align}
 x^{c}_{k+1} - x^{\prime,c}_{k+1} = A_c (x^{c}_{k} - x^{\prime,c}_{k}),
\end{align}
for which the solution can be written as:
\begin{align}
   x^{c}_{k-1} - x^{\prime,c}_{k-1} = A_c^{k-1} (x^{c}_{0} - x^{\prime,c}_{0}).
\end{align}
Defining $\hat{x}_{0 \mid-1}-\hat{x}_{0 \mid-1}^{\prime} \triangleq \zeta, x^{c}_{0} - x^{\prime,c}_{0} \triangleq \kappa$ and substituting the above equation into Eqn.~(\ref{e:diff_real_vir}) gives:
\begin{align}
\nonumber \hat{x}_{k \mid k-1}-\hat{x}_{k \mid k-1}^{\prime}&=\mathcal{A}^{k}\zeta+ \sum_{i=1}^{k} \mathcal{A}^{k-i} \mathcal{B} A_c^{i-1} \kappa \\+ &~~~ \sum_{i=1}^{k} \mathcal{A}^{k-i}  B (\Delta u_i - \Delta u^{\prime}_i). 
\end{align}

Now write the residue as
\begin{align}
\nonumber y_{k}^{\prime}-&C \hat{x}_{k \mid k-1}  =\left(y_{k}^{\prime}-C \hat{x}_{k \mid k-1}^{\prime}\right) - C \mathcal{A}^{k} \zeta \\
\nonumber & -  C\sum_{i=1}^{k} \mathcal{A}^{k-i} \mathcal{B} A_c^{i-1} \kappa - C\sum_{i=1}^{k} \mathcal{A}^{k-i}  B (\Delta u_i - \Delta u^{\prime}_i). 
\end{align}

Notice that $y_{k}^{\prime}-C \hat{x}_{k \mid k-1}^{\prime}$ follows exactly the same distribution as $y_{k}-C \hat{x}_{k \mid k-1}$. Hence, if $\mathcal{A}$ is stable, the second term will converge to 0 and if $\mathcal{A}$ and $A_c$ are both stable, the third term will converge to 0. Thus, $y_{k}^{\prime}-C \hat{x}_{k \mid k-1}$ will converge to the same distribution as $y_{k}-C \hat{x}_{k \mid k-1}$ if there is no watermarking signal added to the system, hence, the detection rate for $\chi^{2}$ detector will be the same as false alarm rate, making the detector useless.

Now notice that $\Delta u_{i}$ is independent of the virtual system and for the virtual system, $y_{k}^{\prime}-C \hat{x}_{k \mid k-1}^{\prime}$ is independent of $\Delta u_{i}^{\prime}$. Hence
\begin{align}
\lim _{k \rightarrow \infty} \operatorname{Cov}\left(y_{k}^{\prime}-C \hat{x}_{k \mid k-1}\right)=\mathcal{X}+2 C {\mathcal{U}} C^{T}, \label{eq:Cov}
\end{align}
where $\mathcal{U}$ is the solution of Lyapunov equation (\ref{eq:calU}). Finally, the expectation of $\chi^{2}$ detector $g_k$ is written as:
\begin{align}
\nonumber &\lim _{k \rightarrow \infty} E\left[\left(y_{k}^{\prime}-C \hat{x}_{k \mid k-1}\right)^{T} \mathcal{X}^{-1}\left(y_{k}^{\prime}-C \hat{x}_{k \mid k-1}\right)\right] \\
\nonumber &~~~~~~=\operatorname{trace}\left[\lim _{k \rightarrow \infty} \operatorname{Cov}\left(y_{k}^{\prime}-C \hat{x}_{k \mid k-1}\right) \times \mathcal{X}^{-1}\right], 
\end{align}
which can be added up over the $T$ time window to prove the results.

\end{document}